\documentclass[12pt]{article}
\usepackage{amsmath}
\usepackage{graphicx,psfrag,epsf,caption}
\usepackage{enumerate}
\usepackage{psfrag,epsf,booktabs}
\usepackage{url} 
\usepackage{amsmath,amssymb,amsfonts,amsthm,mathtools,mathrsfs}
\usepackage[table]{xcolor}

\usepackage{bm,bbm}
\usepackage{color}
\usepackage{subfigure}
\usepackage{algorithm,algpseudocode}
\usepackage[symbol]{footmisc}
\usepackage{scalefnt}
\usepackage{authblk} 
\usepackage{multirow,centernot}
\usepackage[colorlinks=true, citecolor=blue, urlcolor=blue]{hyperref}
\usepackage{natbib}
\usepackage{dsfont}
\usepackage[title]{appendix}
\input cyracc.def

\usepackage[left=1.5in,top=1.1in,right=1in,bottom=1in]{geometry}

\theoremstyle{definition}

\newtheorem{theorem}{Theorem}[section]

\newtheorem{proposition}[theorem]{Proposition}
\newtheorem{remark}[theorem]{Remark}

\makeatletter
\def\@seccntformat#1{\@ifundefined{#1@cntformat}%
	{\csname the#1\endcsname\quad}
	{\csname #1@cntformat\endcsname}
}
\makeatother

\newif\ifShowComments
\ShowCommentstrue
\def\strutdepth{\dp\strutbox}
\def\druk#1{\strut\vadjust{\kern-\strutdepth
        {\vtop to \strutdepth{%
                \baselineskip\strutdepth\vss
                        \llap{\hbox{#1}\quad}\null}}}}

\title{\bf
{\color{black}Closed-form solutions for parameter estimation in exponential families based on maximum a posteriori equations}
}

\author{
\text{Roberto Vila}$^{1}$\thanks{Corresponding author: Roberto Vila, email: {rovig161@gmail.com}
}
\,\,,
\text{Helton Saulo}$^{1,2}$,
\,\,and
\text{Eduardo Nakano}$^{1}$
\\
{\small $^{1}$ Department of Statistics, University of Brasilia, Brasilia, Brazil}\\
{\small $^{2}$ Department of Economics, Federal University of Pelotas, Pelotas, Brazil}\\
}

\begin{document}
	\maketitle 	
	\begin{abstract}
{\color{black}In this paper, we derive closed-form estimators for the parameters of certain exponential family distributions through the maximum a posteriori (MAP) equations.} A Monte Carlo simulation is conducted to assess the performance of the proposed estimators. The results show that, as expected, their accuracy improves with increasing sample size, with both bias and mean squared error approaching zero. Moreover, the proposed estimators exhibit performance comparable to that of traditional MAP and maximum likelihood (ML) estimators. A notable advantage of the proposed  method lies in its computational simplicity, as it eliminates the need for numerical optimization required by MAP and ML estimation.

	\end{abstract}
	\smallskip
	\noindent
	{\small {\bfseries Keywords.} {Exponential family $\cdot$ (MAP) equations $\cdot$ MAP estimator $\cdot$  Maximum likelihood method $\cdot$ Monte Carlo simulation $\cdot$ R software.}}
	\\
	{\small{\bfseries Mathematics Subject Classification (2010).} {MSC 60E05 $\cdot$ MSC 62Exx $\cdot$ MSC 62Fxx.}}
	
	
	\section{Introduction}
	\noindent

{\color{black}
Closed-form estimators are generally computationally efficient, avoiding the convergence issues and high computational cost commonly found in iterative optimization methods. In recent years, several authors have proposed closed-form estimators derived from likelihood-based methods. For instance, \cite{YCh2016} obtained analytical estimators for the gamma distribution by considering the generalized gamma distribution, which results from the transformation $Y=X^{1/\gamma}$ where $X$ is gamma-distributed and $\gamma>0$ is a shape parameter. Similar methodologies have been employed for the Nakagami distribution \citep{RLR2016, Zhao2021}, the weighted Lindley distribution \citep{Kim2020}, and for distributions within the exponential family framework \citep{Vila2024a, Vila2024b}, which includes the previously mentioned cases. Recently, \cite{Kim2022} proposed a novel approach for deriving closed-form estimators by extending the Box-Cox transformation.

Following this line, we propose closed-form estimators for the parameters of a family of probability distributions belonging to the exponential class, with probability density function (PDF) given by
\begin{align}\label{pdf-1}
f(x;\boldsymbol{\psi})
=
{(\mu\sigma)^\mu \over \Gamma(\mu)}\,
{\vert T'(x)\vert} T^{\mu-1}(x)
\exp\left\{-\mu \sigma T(x)\right\},
\end{align}
where $x\in (0,\infty)$, $\boldsymbol{\psi}=(\mu,\sigma)^\top$, with $\mu,\sigma>0$, and $T:(0,\infty)\to (0,\infty)$ is a real, strictly increasing and twice differentiable function, referred to as the generator. Here, $T'(x)$ denotes the derivative of $T(x)$ with respect to $x$. Note that $T(x)$ may involve other known parameters (see Table \ref{table:1} in Appendix \ref{table}).

The proposed estimators are based on maximum a posteriori (MAP) equations, and use the same conceptual framework introduced by \cite{YCh2016} and \cite{Cheng-Beaulieu2002} to develop analytical estimators for certain distributions belonging to the exponential family. The MAP method is a widely used Bayesian approach for deriving point estimates of distribution parameters. Nevertheless, it is uncommon for multi-parameter distributions to yield MAP estimators in closed-form. Therefore, the proposed method may be an alternative to the traditional MAP estimators. In fact, our Monte Carlo results showed that the proposed estimators improved with larger samples, with bias and MSE decreasing, and matched the performance of traditional MAP and maximum likelihood (ML) estimators in the gamma case.

The paper is structured as follows: Sections \ref{The New Estimators}, \ref{Simulation study}, and \ref{Concluding}. Section~\ref{The New Estimators} presents the proposed estimation methodology along with relevant theoretical developments and examples of closed-form estimators. Section~\ref{Simulation study} reports a Monte Carlo simulation study designed to evaluate the performance of the proposed estimators. Finally, Section~\ref{Concluding} concludes the paper with final considerations.
}

\section{The new estimators}\label{The New Estimators}

{\color{black}
By applying the transformation $Y=X^{1/p}$ with $p>0$, as discussed in \cite{YCh2016} and \cite{Cheng-Beaulieu2002}, and assuming that $X$ follows the density in \eqref{pdf-1}, the corresponding PDF of $Y$ takes the form:

\begin{align}\label{dist-gen-exp}
f(y;\boldsymbol{\psi},p)
=
p\, {(\mu\sigma)^\mu \over \Gamma(\mu)}\,
{y^{p-1} \vert T'(y^p)\vert} T^{\mu-1}(y^p)
\exp\left\{-\mu \sigma T(y^p)\right\},
\quad
y\in (0,\infty),
\end{align}
where $\boldsymbol{\psi}=(\mu,\sigma)^\top$ and $\mu,\sigma, p>0$. Table \ref{table:1} in Appendix \ref{table} includes examples of generator functions $T(x)$ that can be used in expressions \eqref{pdf-1} and \eqref{dist-gen-exp}.
}

Let $\{Y_i : i = 1,\ldots , n\}$ be a random sample of size $n$ from  $Y$ having PDF \eqref{dist-gen-exp}, and let $y_i$ be the correspondent observations of $Y_i$. 
The likelihood function for $(\boldsymbol{\psi},p)$ is given by
\begin{align}\label{likelihood function}
L(\boldsymbol{\psi},p\vert \boldsymbol{y})    
=
p^n\, {(\mu\sigma)^{n\mu}  \over \Gamma^n(\mu)}\,
\prod_{i=1}^{n}
y_i^{p-1} \vert T'(y_i^p)\vert T^{\mu-1}(y^p_i)\,
\exp\left\{-\mu \sigma \sum_{i=1}^{n} T(y^p_i)\right\}.  
\end{align}

Let $\pi(\boldsymbol{\psi})$ be the joint density of $\boldsymbol{\psi}=(\mu,\sigma)^\top$ and let $p\sim{\rm Gamma}(\alpha_3,\beta_3)$, for known hyperparameter $\alpha_3$, with density $\pi(p)$. Furthermore, let's suppose that $\boldsymbol{\psi}$ is independent of $p$. Note that in these notations we are committing notational abuse by using $\mu$, $\sigma$ and $p$ for the values of the random variables $\mu$, $\sigma$ and $p$, respectively. This notational abuse is common in Bayesian statistics because it avoids overloading the notation.

The posterior distribution is given by
\begin{align}\label{id-11}
\pi(\boldsymbol{\psi},p\vert \boldsymbol{y})
=
{L(\boldsymbol{\psi},p\vert \boldsymbol{y}) 
	\pi(\boldsymbol{\psi})\pi(p)\over \pi(\boldsymbol{y})},
\end{align}
where $\pi(\boldsymbol{y})\equiv \int_{(0,\infty)^3} L(\boldsymbol{\psi},p\vert \boldsymbol{y}) 
\pi(\boldsymbol{\psi})\pi(p){\rm d}\boldsymbol{\psi} {\rm d}p$ is the predictive distribution, which is finite.

The maximum a posteriori (MAP) estimation method estimates $(\boldsymbol{\psi},p)^\top$, with $\boldsymbol{\psi}=(\mu,\sigma)^\top$, as the mode of the posterior distribution of the random variable $\pi(\boldsymbol{\psi},p\vert \boldsymbol{y})$:	
\begin{align*}
\widehat{(\boldsymbol{\psi},p)}_{\rm MAP}
=
\underset{(\boldsymbol{\psi},p)}{\rm arg\, max}
\pi(\boldsymbol{\psi},p\vert \boldsymbol{y}).
\end{align*}

%

Since $\log(\pi(\boldsymbol{\psi},p\vert \boldsymbol{y}))$ is differentiable over parameter space,
for the occurrence of a maximum (or a minimum) it is necessary that
\begin{align}\label{MAP-equations}
\nabla \log(\pi(\boldsymbol{\psi},p\vert \boldsymbol{y}))=\boldsymbol{0},
\end{align}
where $\nabla=(\partial/\partial\mu,\partial/\partial\sigma,\partial/\partial p)$ is the gradient vector and $\boldsymbol{0}$ is the zero-vector for three-dimensional space. The system in \eqref{MAP-equations} are known as the MAP equations.

\smallskip
For simplicity, from now on we assume the components of the random vector $\boldsymbol{\psi}=(\mu,\sigma)^\top$ are independent and have gamma distributions. That is, the priors $\mu\sim{\rm Gamma}(\alpha_1,\beta_1)$, 
$\sigma\sim{\rm Gamma}(\alpha_2,\beta_2)$ and 
$p\sim{\rm Gamma}(\alpha_3,\beta_3)$, $\alpha_3\neq 1$, are independent, for known hyperparameters $\alpha_k, \beta_k$, $k=1,2,3$.

By taking the logarithm of the posterior probability \eqref{id-11} and by using the formula for $L(\boldsymbol{\psi},p\vert \boldsymbol{y})$ given in \eqref{likelihood function}, we have
\begin{align*}
\log(\pi(&\boldsymbol{\psi},p\vert \boldsymbol{y}))
=
\log(L(\boldsymbol{\psi},p\vert \boldsymbol{y})) 
+
\log(\pi(\mu))
+
\log(\pi(\sigma))
+
\log(\pi(p))
-
\log(\pi(\boldsymbol{y}))
\\[0,3cm]
&=
n\log(p)
+
{n\mu}\log(\mu) 
+ 
{{n\mu} \log(\sigma)
	-
	n\log(\Gamma(\mu))}
-
\log(\pi(\boldsymbol{y}))
\\[0,1cm]
&
+
\sum_{i=1}^{n}
\log(\vert T'(y_i^p)\vert)
+
(p-1)
\sum_{i=1}^{n}
\log(y_i)
-
\mu \sigma \sum_{i=1}^{n} T(y^p_i)
+
(\mu-1)\sum_{i=1}^{n}\log(T(y^p_i))
\\[0,2cm]
&
+
(\alpha_1-1)\log(\mu)
-
\beta_1 \mu
+
(\alpha_2-1)\log(\sigma)
-
\beta_2 \sigma
+
(\alpha_3-1)\log(p)
-
\beta_3 p
\\[0,2cm]
&+
\alpha_1\log(\beta_1)-\log(\Gamma(\alpha_1))
+
\alpha_2\log(\beta_2)-\log(\Gamma(\alpha_2))
+
\alpha_3\log(\beta_3)-\log(\Gamma(\alpha_3)).
\end{align*}
The MAP estimator  of $(\boldsymbol{\psi},p)^\top$ satisfies the MAP equations \eqref{MAP-equations}, that is, 
\begin{align*}
{\partial \log(\pi(\boldsymbol{\psi},p\vert \boldsymbol{Y})) \over\partial \mu}
&=
n\log(\mu)+n\log(\sigma)+n-n\psi^{(0)}(\mu)
\\[0,2cm]
&-
\sigma \sum_{i=1}^{n} T(Y^p_i)
+
\sum_{i=1}^{n}\log(T(Y^p_i))
+
{\alpha_1-1\over \mu}-\beta_1
=
0,  
\\[0,2cm]
{\partial \log(\pi(\boldsymbol{\psi},p\vert \boldsymbol{Y})) \over\partial \sigma}
&=
{n\mu \over \sigma} 
-
\mu \sum_{i=1}^{n} T(Y^p_i)
+
{\alpha_2-1\over\sigma}-\beta_2
=
0,
\\[0,2cm]
{\partial \log(\pi(\boldsymbol{\psi},p\vert \boldsymbol{Y})) \over\partial p}
&=
{n\over p}
+
{1\over p}
\sum_{i=1}^{n}
{T''(Y^p_i)\over T'(Y_i^p)}\, Y^p_i \log(Y_i^p)
+
{1\over p}
\sum_{i=1}^{n}
\log(Y_i^p)
\\[0,2cm]
&
-
{\mu \sigma\over p} \sum_{i=1}^{n} T'(Y^p_i) Y^p_i \log(Y_i^p)
+
{(\mu-1)\over p}\sum_{i=1}^{n} {T'(Y^p_i)\over T(Y^p_i)}\,  Y^p_i \log(Y_i^p)
\\[0,2cm]
&+
{\alpha_3-1\over p}-\beta_3
=
0, 
\end{align*}
where $\psi^{(m)}(x)=\partial^{m+1} \log(\Gamma(x))/\partial x^{m+1}$ is the polygamma function of order $m$. 
Note that the above equations become
\begin{align*}
	{\partial \log(\pi(\boldsymbol{\psi},p\vert \boldsymbol{Y})) \over\partial \mu}
	&=
	n\log(\mu)+ n\log(\sigma)+n-n\psi^{(0)}(\mu)
	-
	\sigma  \sum_{i=1}^{n} T(X_i)
	+ 
	\sum_{i=1}^{n}\log(T(X_i))
	\\[0,2cm]
	&
	+
	{\alpha_1-1\over \mu}-\beta_1
	 =
	0,  
	\\[0,2cm]
	{\partial \log(\pi(\boldsymbol{\psi},p\vert \boldsymbol{Y})) \over\partial \sigma}
	&=
	{n\mu \over \sigma} 
	-
	\mu \sum_{i=1}^{n} T(X_i)
	+
	{\alpha_2-1\over\sigma}-\beta_2
	=
	0,
	\\[0,2cm]
	{\partial \log(\pi(\boldsymbol{\psi},p\vert \boldsymbol{Y})) \over\partial p}
	&=
		\bigg[
	{n}
	+
	\sum_{i=1}^{n}
	{T''(X_i)\over T'(X_i)}\, X_i \log(X_i)
	+
	\sum_{i=1}^{n}
	\log(X_i)
		-
	{\mu \sigma} \sum_{i=1}^{n} T'(X_i) X_i \log(X_i)
	\\[0,2cm]
	&
	+
	{(\mu-1)}\sum_{i=1}^{n} {T'(X_i)\over T(X_i)}\,  X_i \log(X_i)
	+
	{\alpha_3-1}
		\Bigg]
		{1\over p}
	-
	\beta_3
	=
	0, 
\end{align*}
because $Y_i = X_i^{1/p}$, $\forall i=1,\ldots,n$.
Applying expectation over $p$ to both sides of ${\partial \log(\pi(\boldsymbol{\psi},p\vert \boldsymbol{Y})) /\partial p} = 0$ yields
\begin{align*}
&
		\bigg[
{n}
+
\sum_{i=1}^{n}
{T''(X_i)\over T'(X_i)}\, X_i \log(X_i)
+
\sum_{i=1}^{n}
\log(X_i)
\\[0,2cm]
&
-
{\mu \sigma} \sum_{i=1}^{n} T'(X_i) X_i \log(X_i)
+
{(\mu-1)}\sum_{i=1}^{n} {T'(X_i)\over T(X_i)}\,  X_i \log(X_i)
\Bigg]
{\beta_3\over \alpha_3-1}
=
0,
\end{align*}
because $p\sim{\rm Gamma}(\alpha_3,\beta_3)$ and $\mathbb{E}(1/p)={\beta_3/(\alpha_3-1)}$, for $\alpha_3\neq 1$.
Setting $\alpha_3=\beta_3=\kappa$ in the above equation and then letting $\kappa \to \infty$, we obtain
\begin{align*}
	&
	1
	+
	{1\over n}
	\sum_{i=1}^{n}
	{T''(X_i)\over T'(X_i)}\, X_i \log(X_i)
	+
		{1\over n}
	\sum_{i=1}^{n}
	\log(X_i)
	\\[0,2cm]
	&
	-
	{\mu \sigma} \, 	{1\over n} \sum_{i=1}^{n} T'(X_i) X_i \log(X_i)
	+
	{(\mu-1)}\, 	{1\over n}\sum_{i=1}^{n} {T'(X_i)\over T(X_i)}\,  X_i \log(X_i)
	=
	0,
\end{align*}
By resolving the above equation, we can express $\mu$ as a function of $\sigma$ as follows:
\begin{align}\label{sigma-1}
	\mu(\sigma)
	=
	\dfrac{\displaystyle 
		1
		+
		\overline{X}_2
	}{
		\sigma
		\overline{X}_4
		-		
		\overline{X}_3
	},
\end{align}
where
\begin{align}\label{eqq-1}
	\begin{array}{lllll}
&\displaystyle
\overline{X}_2
\equiv 
{1\over n}
\sum_{i=1}^{n}
h_2(X_i),
\quad 
h_2(x)
\equiv 
\log(x)
+
\left[{T''(x)\over T'(x)}-{T'(x)\over T(x)}\right] x \log(x),
\\[0,5cm]
&\displaystyle
\overline{X}_3
\equiv
{1\over n}
\sum_{i=1}^{n} 
h_3(X_i),
\quad 
h_3(x)
\equiv 
{T'(x)\over T(x)}\,  x \log(x),
\\[0,5cm]
&\displaystyle
\overline{X}_4
\equiv
{1\over n}\sum_{i=1}^{n}
h_4(X_i),
\quad 
h_4(x)
\equiv
 T'(x) x \log(x).
	\end{array}
\end{align}
By replacing \eqref{sigma-1} into ${\partial \log(\pi(\boldsymbol{\psi},p\vert \boldsymbol{Y})) /\partial \sigma}= 0$, we obtain the closed-form estimator for $\sigma$ as follows:
%
%
{\small
\begin{align}\label{estimador-mu-new}
	\widehat{\sigma}
	=
	{\displaystyle
		{{1\over n} \, \overline{X}_5}
		-	
		\left(
		{1}
		+
		\overline{X}_2\right)
		\overline{X}_1
		+ 
		\sqrt{
			\left[
			{{1\over n} \, \overline{X}_5}
			-
			\left(
			{1}
			+
			\overline{X}_2\right)
			\overline{X}_1
			\right]^2
			-
			4\, {\beta_2\over n}\, \overline{X}_4	
			\left[
			\left({\alpha_2-1\over n}\right)	
			\overline{X}_3
			-
			\left(
			{1}
			+
			\overline{X}_2
			\right)
			\right]
		}
		\over 
		\displaystyle
		2\,{\beta_2\over n}\,\overline{X}_4
	},
\end{align}
}\noindent
with
\begin{align}\label{eqq-2}
	\begin{array}{lllll}
&\displaystyle
\overline{X}_1
\equiv 
{1\over n}\sum_{i=1}^{n} 
h_1(X_i),
\quad 
h_1(x)
\equiv
T(x),
\\[0,5cm]
&\displaystyle
{\overline{X}_5}
\equiv
{\beta_2}\,
\overline{X}_3
+
\left({\alpha_2-1}\right)\overline{X}_4
=
{1\over n}\sum_{i=1}^{n} 
h_5(X_i),
\quad 
h_5(x)
\equiv
{\beta_2}\,
h_3(x)
+
\left({\alpha_2-1}\right)h_4(x).
	\end{array}
\end{align}

By plugging \eqref{estimador-mu-new} into \eqref{sigma-1}, the closed-form estimator for $\mu$ is given by
\begin{align}\label{estimador-sigma-new}
\widehat{\mu}
=
\dfrac{\displaystyle 
	1
	+
	\overline{X}_2
}{
	\widehat{\sigma}
	\overline{X}_4
	-		
	\overline{X}_3
}.
\end{align}

\begin{remark}
Notice that the new estimators $\widehat{\sigma}$ and $\widehat{\mu}$ are independent of the known hyperparameters $\alpha_1$ and $\beta_1$.
\end{remark}

\begin{proposition}\label{prop-main}
 If $T(x) =x^{-s}$, for some $s\neq 0$, then, the closed-form estimators for  $\sigma$ and $\mu$ are given by
\begin{align}\label{estimador-mu-new-new-special}
	\widehat{\sigma}
=
{\displaystyle
	{{1\over n}\, \overline{X}_5}
	-	
	\overline{X}_1
	+ 
	\sqrt{
		\left(
	{{1\over n}\,\overline{X}_5}
		-
		\overline{X}_1
		\right)^2
		-
		4\, {\beta_2\over n}\, \overline{X}_4	
		\left[
		\left({\alpha_2-1\over n}\right)	
		\overline{X}_3
		-
1
		\right]
	}
	\over 
	\displaystyle
	2\,{\beta_2\over n}\,\overline{X}_4
}
\end{align}
and
\begin{align}\label{estimador-sigma-new-special}
\widehat{\mu}
=
{\displaystyle
	{1}
	\over 
	\widehat{\sigma}
	\overline{X}_4
	-		
	\overline{X}_3
},
\end{align} 
 respectively, where
 \begin{align*}
&\overline{X}_1
=
{1\over n}\sum_{i=1}^{n} X_i^{-s},
\quad
\overline{X}_3
=
{1\over n}\sum_{i=1}^{n} \log(X_i^{-s}),
\quad
\overline{X}_4
=
{1\over n}\sum_{i=1}^{n} X_i^{-s} \log(X_i^{-s}),
\\[0,2cm]
&{\overline{X}_5}
=
{\beta_2}\left[{1\over n}\sum_{i=1}^{n} \log(X_i^{-s})\right]
+
\left({\alpha_2-1}\right)\left[{1\over n}\sum_{i=1}^{n} X_i^{-s} \log(X_i^{-s})\right].
 \end{align*}
\end{proposition}
\begin{proof}
A simple observation shows that, if $T(x) =x^{-s}$, for some $s\neq 0$, then $\overline{X}_2=0$. Hence, from  \eqref{estimador-sigma-new} and \eqref{estimador-mu-new},   $\widehat{\sigma}$ and $\widehat{\mu}$ can be written as in \eqref{estimador-mu-new-new-special} and \eqref{estimador-sigma-new-special}, respectively.
\end{proof}

\begin{remark}
When $\mu = \mu_0$ is constant, the MAP equations yield 
\begin{align*}
	\widehat{\sigma} = \dfrac{\displaystyle {1\over\mu_0}(1+\overline{X}_2)+\overline{X}_3}{\overline{X}_4}. 
\end{align*}
In Table \ref{table:1} of Appendix \ref{table}, examples of distributions that fit within the condition $\mu = \mu_0$ constant are:
Maxwell-Boltzmann, Rayleigh, Weibull, Inverse Weibull (Fréchet), Gompertz, Traditional Weibull, Flexible Weibull, Burr type XII (Singh-Maddala) and Dagum (Mielke Beta-Kappa).

\bigskip 
Likewise, when $\sigma = \sigma_0$ is constant, 
\begin{align*}
	\widehat{\mu}
	=
	\dfrac{\displaystyle 
		1
		+
		\overline{X}_2
	}{
		\sigma_0
		\overline{X}_4
		-		
		\overline{X}_3
	}.
\end{align*}
The only example in Table \ref{table:1} of a distribution that fits the condition  $\sigma = \sigma_0$ constant is the Modified Weibull extension.
\end{remark}

\begin{remark}
	When $\sigma$ is a function of $\mu$, that is $\sigma = g(\mu)$, for some function $g$, the estimator for $\mu$ can be derived from \eqref{estimador-sigma-new} by resolving the following equation:
	\begin{align*}
	\widehat{\mu}g(\widehat{\mu})={1+\overline{X}_3\over \overline{X}_4}.
	\end{align*}
	Examples in Table \ref{table:1} that fit within the condition $\sigma = g(\mu)$ are: $\delta$-gamma and Chi-squared, with $g(x)=1/(\delta x)$, $\delta$ known, and $g(x)=1/(2 x)$, respectively.
\end{remark}

\section{\color{black}Illustrative Monte Carlo simulation study}\label{Simulation study}

{\color{black}We perform two Monte Carlo simulation studies. The first assesses the performance of the proposed closed-form estimators for $\mu$ and $\sigma$ under the gamma, inverse gamma, Weibull, and inverse Weibull distributions. The second focuses on the gamma distribution, comparing the proposed estimators with traditional MAP estimators.} To evaluate the performance of the proposed and traditional MAP estimators, we computed the relative bias and mean squared error (MSE), defined respectively as
\begin{eqnarray*}
 \widehat{\mathrm{Relative\, Bias}}(\widehat{\theta}) =  \frac{1}{N} \sum_{i = 1}^{N} \left|\frac{\widehat{\theta}^{(i)} - \theta}{\theta}\right|, \quad
 \widehat{\mathrm{MSE}}(\widehat{\theta}) = {\frac{1}{N} \sum_{i = 1}^{N} \left(\widehat{\theta}^{(i)} - \theta\right)^2},
\end{eqnarray*}
where $\theta \in \{\mu, \sigma\}$ denotes the true parameter value, and $\widehat{\theta}^{(i)} \in \{\widehat{\mu}^{(i)}, \widehat{\sigma}^{(i)}\}$ represents the $i$th Monte Carlo estimate.  The number of replications is $N = 10,000$. Throughout all simulations, we used the fixed hyperparameter configuration $\alpha_1 = \beta_1=\alpha_2 = \beta_2 = 1/100$, and considered the following sample sizes $n \in \{15, 30, 60, 120, 240, 480, 760\}$. We set $\mu = 2$ and $\sigma = 1$.

\subsection{Performance of the proposed closed-form estimators}

We carried out simulations for assessing the performance of the proposed closed-form estimators of $\sigma$ and $\mu$ under the gamma, inverse gamma, Weibull and inverse Weibull distributions. Note that the estimators of $\mu$ and $\sigma$, based on these distributions, are given by (see  Proposition \ref{prop-main})
\begin{align}\label{estimador-a}
	\widehat{\sigma}
	=
	{\displaystyle
		{{1\over n}\, \overline{X}_5}
		-	
		\overline{X}_1
		+ 
		\sqrt{
			\left(
			{{1\over n}\,\overline{X}_5}
			-
			\overline{X}_1
			\right)^2
			-
			4\, {\beta_2\over n}\, \overline{X}_4	
			\left[
			\left({\alpha_2-1\over n}\right)	
			\overline{X}_3
			-
	1
			\right]
		}
		\over 
		\displaystyle
		2\,{\beta_2\over n}\,\overline{X}_4
	},
\end{align}
\begin{align}\label{estimador-b}
	\widehat{\mu}
	=
	{\displaystyle
		{1}
		\over 
		\widehat{\sigma}
		\overline{X}_4
		-		
		\overline{X}_3
	},
\end{align} 
respectively,
with
\begin{align}\label{estimador-c}
	\begin{array}{lllll}
	& \displaystyle
	\overline{X}_1
	=
	{1\over n}\sum_{i=1}^{n} X_i^{-s},
	\quad
	\overline{X}_3
	=
	{1\over n}\sum_{i=1}^{n} \log(X_i^{-s}),
	\quad
	\overline{X}_4
	=
	{1\over n}\sum_{i=1}^{n} X_i^{-s} \log(X_i^{-s}),
	\\[0,6cm]
	& \displaystyle
	{\overline{X}_5}
	=
	{\beta_2}\left[{1\over n}\sum_{i=1}^{n} \log(X_i^{-s})\right]
	+
	\left({\alpha_2-1}\right)\left[{1\over n}\sum_{i=1}^{n} X_i^{-s} \log(X_i^{-s})\right],
		\end{array}
\end{align}
where the choice of $s$ dictates the appropriate estimator for the parameters of the aforementioned four distributions (see Table \ref{table:2-1}).
%
\begin{center}
	\begin{table}[!htb]
	\centering
	\caption{Some values of $s$ for use in 
		\eqref{estimador-a}, \eqref{estimador-b} and \eqref{estimador-c}.}
	\label{table:2-1}
	\begin{tabular}[t]{ccccc}
		\toprule
		Distribution & Gamma  & Inverse gamma & Weibull & Inverse Weibull  \\
		\midrule
		\rowcolor{gray!10} 
		Value of $s$ & $-1$ & $1$ & $-\delta$  & $\delta$ \\
		\bottomrule
	\end{tabular}
\end{table}
\end{center}

Figure~\ref{fig:mc_proposed} shows the Monte Carlo results. From this figure, we note that, as expected, both relative bias and MSE decrease as the sample size increases for all distributions and parameters. Particularly, the relative biases, in the $\sigma$ case, remain close to zero even for small samples.

\begin{figure}[!ht]
\centering
\includegraphics[width=0.48\textwidth]{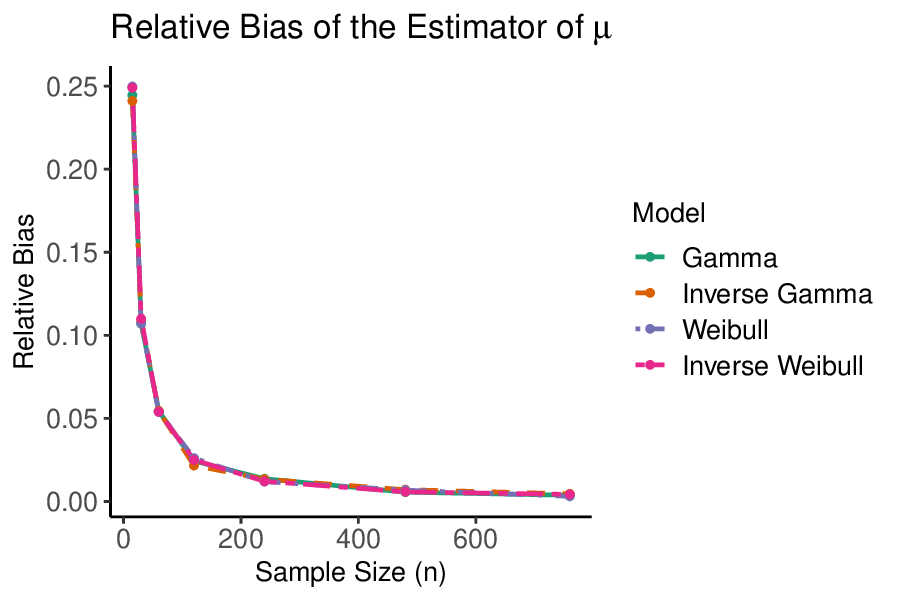}
\includegraphics[width=0.48\textwidth]{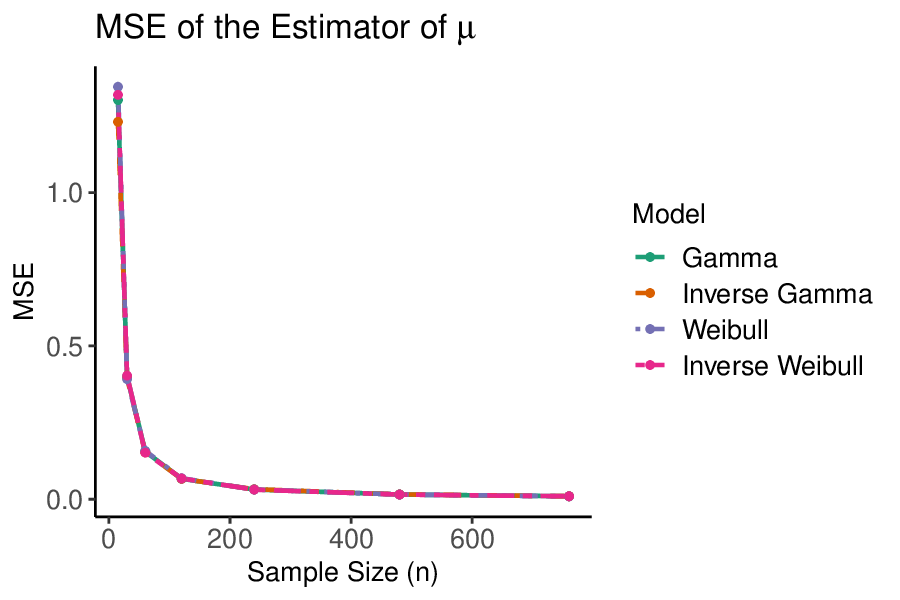}

\vspace{0.4cm}

\includegraphics[width=0.48\textwidth]{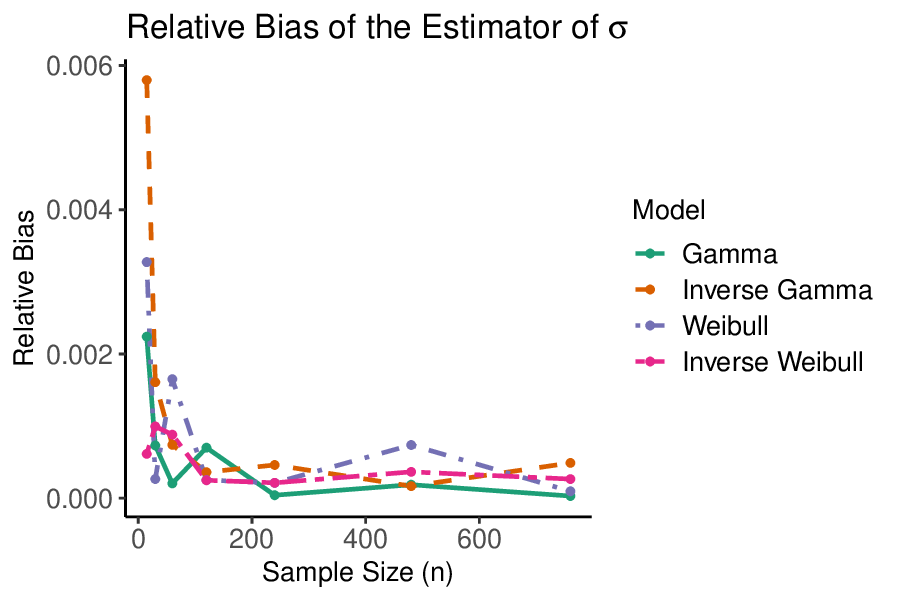}
\includegraphics[width=0.48\textwidth]{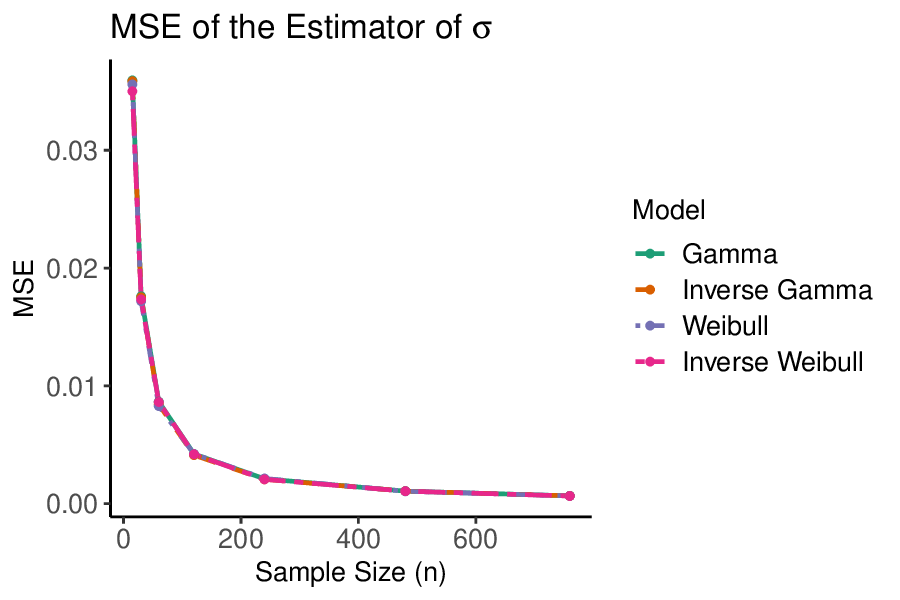}

\caption{Monte Carlo results for the estimators $\widehat{\mu}$ and $\widehat{\sigma}$ for the indicated models. }
\label{fig:mc_proposed}
\end{figure}

\subsection{Comparison with traditional MAP and ML estimators}
We here compare the proposed closed-form estimators with the traditional MAP and ML estimators in the gamma case. Figure~\ref{fig:mle_comparison} shows the relative bias and MSE of both estimators for $\mu$ and $\sigma$ across different sample sizes. As expected, all three estimators improve with increasing $n$. The MAP estimators exhibit slightly lower bias for $\mu$, whereas the ML estimator for $\sigma$ exhibits noticeably higher bias compared to the proposed and MAP estimators. In terms of MSE,  all three approaches present similar performance. Overall, the considered methods provide comparable performance, with the advantage of the proposed estimators lying in their computational simplicity, which avoids numerical optimization.

\begin{figure}[!ht]
    \centering
    \includegraphics[width=0.70\textwidth]{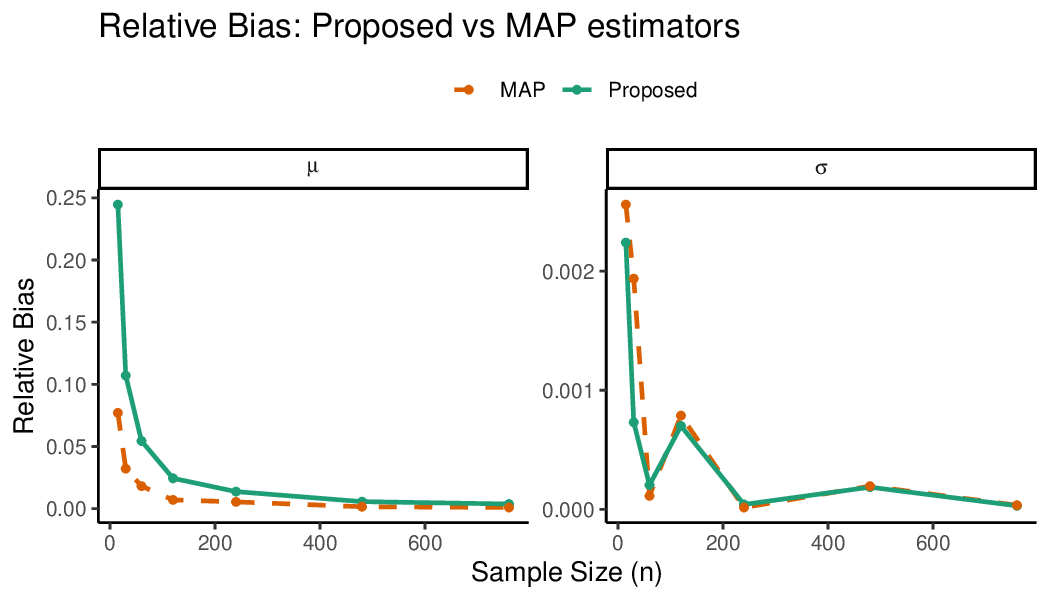}
    \includegraphics[width=0.70\textwidth]{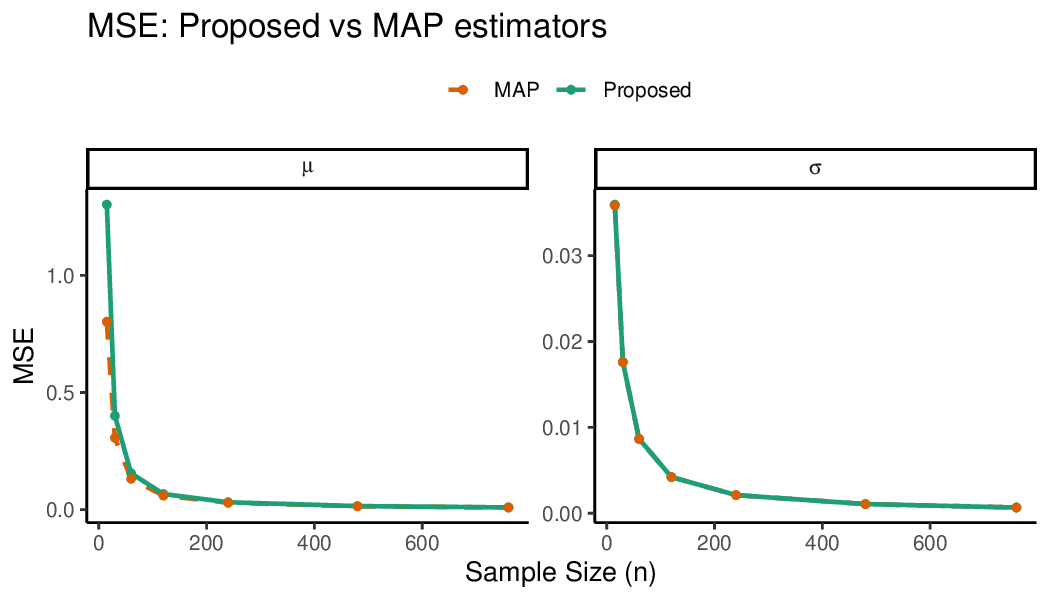}
    \caption{Monte Carlo results comparing the bias and MSE of the proposed, MAP, and ML estimators for the gamma distribution parameters.}
    \label{fig:mle_comparison}
\end{figure}

\section{Concluding remarks}\label{Concluding}

In this paper, we introduced closed-form estimators for a flexible exponential family derived from maximum a posteriori equations. The proposed methodology has as main advantage the elimination of numerical optimization. This feature makes the estimators particularly attractive in applications where computational simplicity is crucial. Monte Carlo simulations revealed that, as expected, the performance of the proposed estimators improved with increasing sample size, with both bias and mean squared error tending toward zero. In the gamma case, the proposed estimators demonstrated comparable performance with the traditional maximum a posteriori and maximum likelihood estimators. Future research may focus on extending the proposed methodology to multivariate distributions and exploring additional distribution families. These investigations are underway, and we hope to report the results in future.

	\paragraph{Acknowledgements}
Funding for this research was provided in part by CAPES, Brazil (Finance Code 001).
	
	\paragraph{Disclosure statement}
The authors declare that they have no conflict of interest.




\newpage
\begin{appendices}
\section{Examples of generators $T$}\label{table}

	\begin{table}[!h]
		\centering
		\caption{Some forms of generators $T(x)$ for use in \eqref{pdf-1}.}
		\label{table:1}
		\begin{tabular}[t]{lccc}
			\toprule
			Distribution & $\mu$ & $\sigma$ & $T(x)$  \\
			\midrule
			\rowcolor{gray!10} 
					Nakagami \citep{Laurenson1994}
			&  $m$  & ${1\over \Omega}$   & $x^2$ 
			\\[0.5ex]
				Maxwell-Boltzmann \citep{Dunbar1982}
&  ${3\over 2}$  & ${1\over 3\beta^2}$   & $x^2$ 
			\\[0.5ex]
			\rowcolor{gray!10} 
				Rayleigh \citep{Rayleigh1880}
&  $1$  & ${1\over 2\beta^2}$   & $x^2$ 
			\\[0.5ex]
				Gamma \citep{Stacy1962}
&  $\alpha$  &  ${1\over\alpha\beta}$ & $x$ 
			\\[0.5ex]
			\rowcolor{gray!10} 
				Inverse gamma \citep{Cook2008}
&  $\alpha$  &  ${1\over\alpha\beta}$ & ${1\over x}$ 
			\\[0.5ex]
				$\delta$-gamma \citep{Rahman2014}
&  ${\beta\over \delta}$  &  ${1\over \beta}$ & $x^\delta$ 
			\\[0.5ex]
			\rowcolor{gray!10} 
				Weibull \citep{Johnson1994}		
&  $1$ & ${1\over \beta^\delta}$ & $x^\delta$
			\\[0.5ex]
				Inverse Weibull (Fréchet) \citep{khan2008}
& $1$ & ${1\over \beta^\delta}$ & ${1\over x^\delta}$ 
			\\[0.5ex]
			\rowcolor{gray!10} 
				Generalized gamma \citep{Stacy1962}
& ${\alpha\over \delta}$ & ${\delta\over \alpha \beta^\delta}$ & $x^\delta$ 
			\\[0.5ex]
				Generalized inverse gamma \citep{Lee1991}
&  ${\alpha\over \delta}$ &  ${\delta\over \alpha\beta^{\delta}}$ & ${1\over x^\delta}$
			\\[0.5ex]
			\rowcolor{gray!10} 
				New log-generalized gamma
& ${\alpha\over \delta}$ & ${\delta\over \alpha \beta^\delta}$ & $[\exp(x)-1]^\delta $
			\\[0.5ex]
				New log-generalized inverse gamma
& ${\alpha\over \delta}$ & ${\delta\over \alpha \beta^\delta}$ & $\big[\exp\big({1\over x}\big)-1\big]^\delta $ 
			\\[0.5ex]
			\rowcolor{gray!10} 
				New exponentiated generalized gamma
& ${\alpha\over \delta}$ & ${\delta\over \alpha \beta^\delta}$ & $\log^\delta(x+1) $
			\\[0.5ex]
			New exponentiated generalized inverse gamma
& ${\alpha\over \delta}$ & ${\delta\over \alpha \beta^\delta}$ & $\log^\delta\big({1\over x}+1\big) $ 
			\\[0.5ex]
						\rowcolor{gray!10} 
				New modified log-generalized gamma
& ${\alpha\over \delta}$ & ${\delta\over \alpha \beta^\delta}$ & $\exp^\delta\big(x-{1\over x}\big) $ 
			\\[0.5ex]
				New extended log-generalized gamma
& ${\alpha\over \delta}$ & ${\delta\over \alpha \beta^\delta}$ & $x^\delta[\exp(x)-1]^\delta$ 
			\\[0.5ex]
						\rowcolor{gray!10} 
				Chi-squared \citep{Johnson1994}	
& ${\nu\over 2}$ & ${1\over \nu}$ & $x$ 
			\\[0.5ex]
				Scaled inverse chi-squared  \citep{Bernardo1993}
& ${\nu\over 2}$ & $\tau^2$ &  ${1\over x}$ 
			\\[0.5ex]
						\rowcolor{gray!10} 
				Gompertz \citep{Gompertz1825}
& $1$ & $\alpha$ &  $\exp(\delta x)-1$ 
			\\[0.5ex]
				Modified Weibull extension \citep{Xie2022}
& $\lambda\alpha$ & $1$ &  $\exp\big[ \left({x\over\alpha}\right)^\beta\big]-1$ 
			\\[0.5ex]
						\rowcolor{gray!10} 	
			Traditional Weibull \citep{Nadarajah2005}
& $1$ & $a$ &  $x^b[\exp( cx^d)-1]$ 
			\\[0.5ex]
				Flexible Weibull \citep{Bebbington2007}
& $1$ & $a$ &  $\exp\left(b x-{c\over x}\right)$
			\\[0.5ex]
						\rowcolor{gray!10} 
				Burr type XII (Singh-Maddala) \citep{Burr1942}
& $1$ & $k$ &  $\log(x^c+1)$ 
			\\[0.5ex]
				Dagum (Mielke Beta-Kappa) \citep{Dagum1975}
& $1$ & $k$ &  $\log\big({1\over x^c}+1\big)$
			\\[0.5ex] 	
			\bottomrule
		\end{tabular}
	\end{table}
In Table \ref{table:1}, we are assuming that $m\geqslant {1/2}$ and $\Omega$, $a,b,c,d,k$, $\alpha,\delta,\lambda,\beta, \nu,\tau^2>0$.

%
%

\end{appendices}
\end{document}